\documentclass[10pt,a4paper,onecolumn]{article}

\usepackage{microtype}
\usepackage{science}
\usepackage{amsmath}
\usepackage{amssymb}
\usepackage{wasysym}
\usepackage[italian,english]{babel}
\usepackage{amsthm}
\usepackage{amscd}
\pdfoutput=1 

\usepackage[colorlinks,hyperindex,plainpages=false]{hyperref}
\usepackage{hypcap}
\usepackage{thumbpdf}
\usepackage{url}

\usepackage{microtype}
\usepackage{science}
\usepackage{amsmath}
\usepackage{latexsym}
\usepackage{amssymb}
\usepackage{wasysym}
\usepackage[italian,english]{babel}
\usepackage{amsthm}
\usepackage{amscd}
\usepackage{enumitem}
\usepackage{cases}
\pdfoutput=1

\theoremstyle{plain}
\newtheorem{thm}{Theorem}[section]

\theoremstyle{definition}
\newtheorem{defn}{Definition}[section]

\newtheorem{exmp}{Example}[section]

\theoremstyle{remark}
\newtheorem*{rem}{Remark}

\usepackage{thumbpdf}

\begin{document}
\title{\center\Large\bf Aphids, Ants and Ladybirds: a mathematical model predicting their population dynamics}
\author{{\bf Gianluca Gabbriellini} \\
        Osnago, 23875, Italy\\
        {\tt gianluca.gabbriellini@gmail.com}
}
\maketitle
\thispagestyle{empty}

\begin{abstract}

The interaction between aphids, ants and ladybirds has been investigated from an ecological point of view since many decades, while there are no attempts to describe it from a mathematical point of view.
This paper introduces a new mathematical model to describe the within-season population dynamics in an ecological patch of a system composed by aphids, ants and ladybirds, through a set of four differential equations. The proposed model is based on the Kindlmann and Dixon set of differential equations \cite{Kindlmann2003}, focused on the prediction of the aphids-ladybirds population densities, that share a prey-predator relationship. The population of ants, in mutualistic relationship with aphids and in interspecific competition with ladybirds, is described according to the Holland and De Angelis mathematical model \cite{HollandDeAngelis2010}, in which the authors faced the problem of mutualistic interactions in general terms. The set of differential equations proposed here is discretized by means the Nonstandard Finite Difference scheme, successfully applied by Gabbriellini to the mutualistic model \cite{gabbriellini2014}. 
The constructed finite-difference scheme is positivity-preserving and characterized by four nonhyperbolic steady-states, as highlighted by the phase-space and time-series analyses. Particular attention is dedicated to the steady-state most interesting from an ecological point of view, whose asymptotic stability is demonstrated via the Centre Manifold Theory. The model allows to numerically confirm that mutualistic relationship effectively influences the population dynamic, by increasing the peaks of the aphids and ants population densities. Nonetheless, it is showed that the asymptotical populations of aphids and ladybirds collapse for any initial condition, unlike that of ants that, after the peak, settle on a constant asymptotic value.

\end{abstract}

\section{Introduction}

It is well known that aphids and ants share a symbiotic relationship: ants collect honeydew produced by aphids using it as food and the ants reciprocate by protecting aphids against predators \cite{novgorodova2005,muller2009}. In such cases the symbiotic relationship is defined \textit{mutualism} \cite{Boucher1982}. Holland and De Angelis in 2010 built a mathematical model that links the consumer functional responses of a mutualistic species with resources supplied by another; through the phase-plane analysis they shown that their set of differential equations correctly predicts the enhanced population growth rates of both species sharing the mutualistic relationship. Although in this paper the relationship between aphids and ants was hypothesized as mutualistic, an alternative and interesting vision is given in \cite{palsson2002}.

One of the most aggressive predators threatening the aphids colonies is the ladybird \cite{novgorodova2005}, in both larval and adult stages. The voracity of the ladybirds has left researchers thinking to use them as a biological control agent against the proliferation of the aphids, since they constitute a problem for crops \cite{VanHemden2017}. Although there is a large scientific community which considers the contribution of the ladybirds to be effective in eradication of aphids, several authors demonstrated through mathematical models supported by experimental tests that they are ineffective \cite{Kindlmann2003,Kindlmann2015}. For a complete discussion and the full bibliography available on this topic refer to \cite{Kindlmann2015}. The mechanism regulating the population dynamics of ladybirds and aphids has been described by Kindlmann and Dixon in 2003 \cite{Kindlmann2003}. The authors proposed a population dynamical model that incorporates an optimization of egg distribution, offering an explanation as to why ladybirds have little effect on the aphids population dynamics. 

The relationship between aphids, ants and ladybirds has been known and studied from a an ecological \cite{oliver2008} and chemicals point of view \cite{pasteels2007}, while a mathematical model that describes this dynamical system it seems not to have been developed yet.

This paper introduces a new set of differential equations describing the within-season population dynamics of aphids, ladybirds and ants. It includes the model proposed by Kindlmann in 2003 to reproduce the within-season aphids-ladybirds interaction and extends it by adding the mutualistic aphids-ants relationship, following the model proposed by Holland and De Angelis in \cite{HollandDeAngelis2010}. 

Starting from the continuous-time model, a discrete-time version was proposed by applying the Non Standard Finite Difference (NSFD) scheme \cite{mickens89}. This mathematical framework allows to numerically integrate nonlinear differential equations of interest for a large variety of scientific fields, such as equations modeling the stellar structure, the dynamics of HIV transmission, heat transport equations and many other topics \cite{mickens05a}. The high reliability of the NSFD respect to standard finite difference schemes proved in several works \cite{mickens89,mickens94,mickens05a,mickens05}. Following the NSFD rules, the continuous-time model proposed by Holland and De Angelis in 2010 has been converted into a set of difference equations, proving that the NSFD performs better respect to a standard finite difference approach during the transient dynamics, especially with high time steps \cite{gabbriellini2014}; since the mutualistic interaction has a fundamental role in the model proposed in this paper, the discretization measures that have been successfully applied in \cite{gabbriellini2014} have been adopted also in this study.

\section{Aphids, ants and ladybirds: the continuous-time model}\label{sec:cont_model}
In this section a set of differential equations describing the within-season population dynamics of a system composed by aphids, ants and ladybirds living in an ecological patch is proposed. The system, that for conciseness reasons can be cited later as AAL, is the following:
\begin{subnumcases}{\label{eq:AAL}}
{\label{eq:AALa}}\frac{dh(t)}{dt}=ax(t) \\ 
{\label{eq:AALb}}\frac{dx(t)}{dt}=(r_1-h(t))x(t)-\frac{\nu p x(t)y(t)}{b+p x(t)+y(t)}+\alpha_{12}\frac{x(t)y(t)}{\epsilon_1+y(t)} \\
{\label{eq:AALc}}\frac{dy(t)}{dt}=r_2y(t)+\alpha_{21}\frac{x(t)y(t)}{\epsilon_2+x(t)}-d_2y^2(t) \\
{\label{eq:AALd}}\frac{dz(t)}{dt}=-\frac{\nu z^2(t)}{b+px(t)+z(t)}-ky(t)z(t),
\end{subnumcases}

with initial conditions $h(0)=0,\,x(0)=x_0\ge 0, \,y(0)=y_0\ge 0\,z(0)=z_0\ge 0$. The (\ref{eq:AAL}) is a system of four non-linear differential equations regulating $h(t),x(t),y(t),z(t)$, respectively the cumulative density of the aphids, densities of aphids, ants and ladybirds at time $t$. Meaningful values of $h(t),x(t),y(t),z(t)$ are non-negative. 

The parameter $a$ is the scaling constant relating aphids cumulative density to its own dynamics; $r_1$ and $r_2$ are respectively the maximum potential growth rate of the aphids and ants; $\nu$ is the ladybirds voracity; $p$ is the ladybirds preference for aphids; $b$ is a parameter of the functional response of the aphids; $\alpha_{12}$ is a positive term quantifying the advantages that the presence of ants induces on the growth of aphids; $\alpha_{21}$ is a positive term quantifying the advantages that the presence of aphids induces on the growth of ants; $\epsilon_1$ and $\epsilon_2$ are the half-saturation terms for aphids and ants respectively; $k$ is the coefficient of interspecific competition between ants and ladybirds; $d_2$ is the self-limiting term of ants population. The initial condition $z_0$ is defined by the number of eggs laid there by adults \cite{Kindlmann2003}.
In order to clarify how the (\ref{eq:AAL}) works, each equation is described below: 
\begin{itemize}
\item The formulation given in (\ref{eq:AALa}) expresses the cumulative density $h(t)$ as a regulatory term for aphids, instead of the instantaneous density \cite{Kindlmann2003}.
\item The (\ref{eq:AALb}) describes changes in aphids density through the sum of three contributes: the first is an auto-regulatory term allowing aphids population density to decline even in the absence of natural enemies \cite{Kindlmann2003}; the second expresses the decrease of the aphids population, related to predator voracity $\nu$ and predator's preference for prey $p$ \cite{Kindlmann2003}; the last term, modulated by the coefficient $\alpha_{12}$, quantifies the advantages that the presence of ants induces on the growth of the aphids, according to the research of Holland and De Angelis about the mutualistic relationships \cite{HollandDeAngelis2010}.
\item The (\ref{eq:AALc}) reproduces the ants population density, in which the first and last terms represent respectively the linear effect due to the intrinsic population growth rate and a quadratic one to modify the growth rate with density dependent self–limitation; the second term, regulated by the coefficient $\alpha_{21}$, quantifies the advantages that the presence of aphids induces on the growth of the ants.
\item The (\ref{eq:AALd}) describes the decreases of the ladybirds population caused by two terms: a first reproducing the cannibalistic effect \cite{Kindlmann2003} and a second one expressing the competition between ladybirds and ants \cite{Majerus}.
\end{itemize}

The analysis of the continuous-time model was not reported in the paper due to the achievement of partial results. The complexity of the proposed set of differential equations makes the effort to carry out the stability analysis of the steady-states remarkable, and suitable to be dealt in a separate work.

\section{Discrete-time model}

To carry out this step, the continuous variable $t\in [0,\infty)$ must be replaced by the discrete variable $n\in\,\mathbb{N}$ and the variables $h(t),x(t),y(t)$ and $z(t)$ must take discrete values $h_n,x_n,y_n$ and $z_n$. The result is a set of difference equations. 

The numerical integration of ordinary differential equations using traditional methods could produce different solutions from those of
the original ODE \cite{cresson16,mickens05}. In particular, using a discretization step-size larger than some relevant time
scale, is possible to obtain solutions that may not reflect the dynamics of the original system. To overcome this problem, Ronald
Mickens, in 1989, suggested what is known as the Nonstandard Finite Difference (NSFD) method \cite{mickens89}. 

As introduced, in the AAL model the mutualistic interaction has an important role. Since in the paper \cite{gabbriellini2014} the NSFD scheme best has performed respect to a standard finite difference scheme, in the present research the author has followed the same way, extending the approach also for the other terms of the AAL dynamical system. 

\subsection{Non-Standard Finite Difference schemes (NSFD)}\label{sub:nsfd_rules}

Let $f: \mathbb{R}^n\longrightarrow \mathbb{R}^n$ sufficiently smoothed, and $\xi(t):[0,+\infty)\longrightarrow \mathbb{R}^n$ the coordinates. Given the differential equation 
\begin{align}\label{eq:diff_eq_app}
\frac{d\xi(t)}{dt}=f(\xi(t),K),
\end{align}
with initial condition $\xi_0=\xi(t=t_0)\in\mathbb{R}^n$ and the system parameters identified by $K=(K_1,K_2,\ldots)$.
We also suppose that 
\begin{align}\label{eq:FD_app}
\frac{\xi_{n+1}-\xi_n}{\Delta t}=F(\xi_n,K),
\end{align}
is the difference equation corresponding to \eqref{eq:diff_eq_app}.
In order to construct a NSFD scheme, the following rules have to be respected \cite{mickens05}:
\begin{enumerate}
\item[I.] The order of the discrete derivative should be equal to the order of the corresponding derivatives of the differential equation. 
\item[II.] Denominator functions for the discrete representation must be nontrivial. The following replacement is then required: 
\begin{align}
\Delta t\longrightarrow\phi(\Delta t)+O(\Delta t^2),
\end{align}
where $\phi(\Delta t)$ is such that $0<\phi(\Delta t)<1$, $\forall\Delta t>0$. The explicit form of $\phi(\Delta t)$ is given by the following expression \cite{mickens00}:
\begin{equation}\label{eq:phi}
\phi(\Delta t)=\frac{1-e^{-q\Delta t}}{q}.
\end{equation}
Let the Jacobian $J_{\Phi}(\xi)=(j_{ij})_{4\times 4}$, with $\xi=(h_n,x_n,y_n,z_n)$ and
\begin{align}
j_{ij}=\frac{\partial \Phi_i}{\partial\xi_j},
\end{align} 
let $\Omega$ the set of the eigenvalues of the Jacobian, the optimal value of $q$ must respect the condition
\begin{equation}\label{eq:condition}
q\ge\max_{\Omega}\Big\{\frac{\lambda^2}{2|\Re(\lambda)|}\Big\}\,\,\,\rm if\,\,\,  \Re(\lambda)\neq 0\,\, for    \,\,\,\lambda\in\Omega.
\end{equation}
\item[III.] Nonlinear terms could be replaced by nonlocal discrete representations.
\item[IV.] Special conditions that hold for the solutions of the differential equations should also hold for the solutions of the finite difference scheme.
\item[V.] The scheme should not introduces spurious solutions.
\end{enumerate}

\begin{defn}[PESN condition] A NSFD scheme preserving the solution positivity is called \textit{Positive and Elementary Stable Nonstandard} (PESN) method. 
\end{defn}
The positivity condition is particularly advantageous to discretize dynamical systems describing a population evolution, being a population density never negative. 

\subsubsection{Nonlocal representation and positivity}
In order to construct a positivity-preserving scheme, it is required that the nonlocal representation of the nonlinear terms (condition III) is chosen carefully. To the author's knowledge, there are no attempts to formalize the way in which the nonlocal approximation are constructed, leaving this fundamental step to the experience. In this paper, a simple rule was proposed.
 
\begin{thm}\label{thm:pos}
Let the differential equation \eqref{eq:diff_eq_app}, $g(\xi,K)$ a function depending in nonpolynomial way on $\xi$, always positive for all $\xi\in\mathbb{R}^n$, and $K_1,K_2,\ldots\in\mathbb{R}^+$. If 
\begin{align}\label{eq:forma_standard}
f(\xi,K)=(-1)^p\xi^m g(\xi,K),
\end{align}
with $m\in\mathbb{N}^+$ and $p=0,1$. Then, a sufficient condition to have a NSFD respecting the positivity condition is:
\begin{align}\label{eq:pos_cond_app}
F(\xi_n,K)=\frac{(-1)^p}{2}\bigg[\big(1-(-1)^{p+1}\big)\xi_n^m+\big(1-(-1)^p\big)\xi_n^{m-1}\xi_{n+1}\bigg]g_n(\xi_n,K).
\end{align}
\end{thm}

\begin{proof}
By substituting the \eqref{eq:pos_cond_app} in \eqref{eq:FD_app},
\begin{align}
\frac{\xi_{n+1}-\xi_n}{\phi(\Delta t)}=\frac{(-1)^p}{2}\bigg[\big(1-(-1)^{p+1}\big)\xi_n^m+\big(1-(-1)^p\big)\xi_n^{m-1}\xi_{n+1}\bigg]g_n(\xi_n,K),
\end{align}
\begin{align}
\xi_{n+1}=\frac{\frac{(-1)^p}{2}g_n(\xi_n,K)\phi(\Delta t)\big(1-(-1)^{p+1}\big)\xi_n^m+\xi_n}{1-\frac{(-1)^p}{2}\xi_n^{m-1}\big(1-(-1)^{p}\big)g_n(\xi_n,K)\phi(\Delta t)}.
\end{align}
It is straightforward to show that
\begin{align}
\xi_{n+1}(p=0)=g_n(\xi_n,K)\phi(\Delta t)\xi_n^m+\xi_n,
\end{align}
and
\begin{align}
\xi_{n+1}(p=1)=\frac{\xi_n}{1+g_n(\xi_n,K)\phi(\Delta t)\xi_n^{m-1}}.
\end{align}
Being $\xi_{n+1}(p=0)\ge \xi_n$ and $\xi_{n+1}(p=1)\le \xi_n$, it follows that
\begin{align}\label{eq:th}
0 \le \xi_{n+1}(p=1)\le \xi_{n+1}(p=0).
\end{align}
The \eqref{eq:th} shows the positivity of the solution.
\end{proof}

\begin{rem}
The rule \eqref{eq:pos_cond_app} allows to build a \textit{minimalist} positivity-preserving finite difference scheme. More advanced nonlocal representations may be required \cite{mickens05}.  
\end{rem}

\begin{exmp}
Let the differential equation
\begin{align}\label{eq:example}
\frac{dx(t)}{dt}=-\frac{x^3}{x^2+\theta}.
\end{align}
with $\theta\in\mathbb{R}^+$. By comparing \eqref{eq:example} with \eqref{eq:forma_standard}, it follows that $g(x,\theta)=\frac{1}{x^2+\theta}$, $p=1$ and $m=3$. Following the \eqref{eq:pos_cond_app}, it result
\begin{align}
F(x_n,\theta)=-\frac{1}{2}\bigg[0+2x_n^2x_{n+1}\bigg]\frac{1}{x_n^2+\theta}=-\frac{x_n^2x_{n+1}}{x_n^2+\theta}.
\end{align}
The NSFD scheme is:
\begin{align}
\frac{x_{n+1}-x_n}{\phi(\Delta t)}=-\frac{x_n^2x_{n+1}}{x_n^2+\theta},
\end{align} 
Then, with a simple algebraic manipulation
\begin{align}
x_{n+1}=\frac{x_n(x_n^2+\theta)}{x_n^2\big(1+\phi(\Delta t)\big)+\theta},
\end{align}
that is a positivity-preserving scheme. 
\end{exmp}

\subsection{AAL: the discrete-time model}

The finite difference scheme of the AAL dynamical system was built by extending the work carried out in \cite{gabbriellini2014}, in which a stable NSFD scheme for a system of two populations in mutualistic relationship has been proposed. In that work, the proposed NSFD scheme respected the rules listed in the Subsection \ref{sub:nsfd_rules}, including the nonlocal approximation to ensure the positivity of the solutions (although not yet formalized). For the AAL, the same approach was followed, by extending the NSFD rules also to the non-mutualistic terms. The set of difference equations proposed is the following:

\begin{equation}\label{eq:AAL_timedisc}
\left\{
\begin{array}{l}
\frac{h_{n+1}-h_n}{\phi(\Delta t)}=ax_n \\

\frac{x_{n+1}-x_n}{\phi(\Delta t)}=r_1 x_n-h_nx_{n+1}-\frac{\nu p x_{n+1}z_n}{b+px_n+z_n}+\alpha_{12}\frac{x_ny_n}{\epsilon_1+y_n} \\ 

\frac{y_{n+1}-y_n}{\phi(\Delta t)}=r_2y_n+\alpha_{21}\frac{x_ny_n}{\epsilon_2+x_n}-d_2y_ny_{n+1} \\

\frac{z_{n+1}-z_n}{\phi(\Delta t)}=-\frac{\nu z_nz_{n+1}}{b+px_n+z_n}-ky_nz_{n+1}, \\
\end{array}
\right.
\end{equation}
with the conditions $h_0=0,\, x_0\ge 0, \,y_0\ge 0\,z_0\ge 0$. Each right term of the equations \ref{eq:AAL_timedisc} are obtained by taking advantage from \eqref{eq:pos_cond_app}. The \eqref{eq:AAL_timedisc} were explicated respect to $h_{n+1},x_{n+1},y_{n+1},z_{n+1}$ as follows:
\begin{equation}\label{eq:AAL_timedisc_expl}
\left\{
\begin{array}{l}
h_{n+1}=\Phi_1(h_n,x_n,y_n,z_n) \\

x_{n+1}=\Phi_2(h_n,x_n,y_n,z_n) \\

y_{n+1}=\Phi_3(h_n,x_n,y_n,z_n) \\

z_{n+1}=\Phi_4(h_n,x_n,y_n,z_n) \\
\end{array}
\right.
\end{equation}
in which
\begin{equation}\label{eq:AAL_timedisc_expl_1}
\left\{
\begin{array}{l}
\Phi_1(h_n,x_n,y_n,z_n)= a x_n \phi(\Delta t) +h_n\\

\Phi_2(h_n,x_n,y_n,z_n)= \frac{x_n (b+p x_n+z_n) \big(\left(\epsilon_1+y_n\right) \left(r_1 \phi(\Delta t)+1\right)+\alpha _{12} y_n \phi(\Delta t) \big)}{\left(\epsilon_1+y_n\right) \big(\phi(\Delta t) (h_n (b+p x_n+z_n)+\nu  p z_n)+b+p x_n+z_n\big)}\\

\Phi_3(h_n,x_n,y_n,z_n)= \frac{y_n \left(\epsilon_2+x_n\right) \left(r_2 \phi +1\right)+\alpha _{21} x_n y_n \phi(\Delta t) }{\left(\epsilon_2+x_n\right) \left(d_2 y_n \phi(\Delta t) +1\right)}\\

\Phi_4(h_n,x_n,y_n,z_n)= \frac{z_n (b+p x_n+z_n)}{\phi(\Delta t)\big(k y_n (b+p x_n+z_n)+\nu z_n\big)+b+p x_n+z_n}.\\
\end{array}
\right.
\end{equation}
All the difference equations of \eqref{eq:AAL_timedisc_expl} are positive, being positive all the parameters and variables.

\subsection{Steady-state analysis}\label{subsec:steady-states}

To find the steady-states of \eqref{eq:AAL_timedisc_expl_1} the
following condition has to be satisfied:
\begin{equation}\label{eq:AAL_steady_states}
\left\{
\begin{array}{l}
\Phi_1(h_n,x_n,y_n,z_n)= h_n\\

\Phi_2(h_n,x_n,y_n,z_n)= x_n\\

\Phi_3(h_n,x_n,y_n,z_n)= y_n\\

\Phi_4(h_n,x_n,y_n,z_n)= z_n\\
\end{array}
\right.
\end{equation}
By analytically solving the \eqref{eq:AAL_steady_states}, the following steady-states $\Gamma=(\tilde h, \tilde x,\tilde y,\tilde z)$ were found: 
\begin{align}
\Gamma^1 = (\tilde h_1, 0,0,0)\,\,\,and\,\,\,\Gamma^2 = \bigg(\tilde h_2, 0,\frac{r_2}{d_2},0\bigg).
\end{align}
The first equation in \eqref{eq:AAL_steady_states} implies that only $\tilde x_n=0$ is admissible for $\phi(\Delta t)>0$, while it gives no information about the asymptotic values $\tilde h_1$ and $\tilde h_2$. In order to estimate them, and complete the knowledge about $\Gamma^1$ and $\Gamma^2$, a numerical simulation was carried out by assigning the following values to the parameters: $r_1=r_2=0.3$, $a=0.000005$, $k=0.005$, $\nu=1$, $b=0$ and $p=1$ according to \cite{Houdkova2006}; $\alpha_{12}=\alpha_{21}=0.6$, $d_2=0.01$ and $\epsilon_1=\epsilon_2=0.3$ according to \cite{HollandDeAngelis2010, gabbriellini2014}; the discretization parameters are $dt=10^{-5}$ and $q=1.5$. After $10^7$ iterations, by assuming the initial values $(h_0,x_0,y_0,z_0)=(0,5,0,1)$, it results $\tilde h_1\simeq 0.599525$; let $(h_0,x_0,y_0,z_0)=(0,100,40,40)$, it results $\tilde h_2\simeq 1.796894$. Also, by varying the initial conditions to $(h_0,x_0,y_0,z_0)=(0,0,0,1)$ a new fixed point $\Gamma^0=(0,0,0,0)$ arises. Finally, with $(h_0,x_0,y_0,z_0)=(0,0,40,40)$, a fourth steady-state $\Gamma^3=(0,0,30,0)$ was obtained. 

In order to know the nature of the four steady-states $\Gamma^0-\Gamma^3$, the eigenvalues of the Jacobian matrix need to be calculated. 
\begin{defn}\label{def:nonhyp}
Given $\sigma(J_{\Phi})$ the set of the Jacobian eigenvalues, a steady-state $\tilde\xi\in\mathbb{R}^n$ is defined: 
\begin{itemize}[noitemsep,topsep=1pt]
\item \textit{hyperbolic} if $\nexists\lambda\in\sigma(J_{\Phi})$ so that $|\lambda|=1$, 
\item \textit{nonhyperbolic} if $\exists_1\lambda\in\sigma(J_{\Phi})$ so that $|\lambda|=1$. 
\end{itemize}
\end{defn}
For each steady-state the eigenvalues are summarized in Table 1 and represented in the complex plane of \figurename~\ref{fig:Eigenvalues}: the eigenvalues characterized by the same symbol are relative to the same steady-state. Since all the four steady-states have a unitary eigenvalue, following the Definition \ref{def:nonhyp} they are nonhyperbolic. 

\begin{figure}[!ht]
\centering
\includegraphics[width=0.6\textwidth]{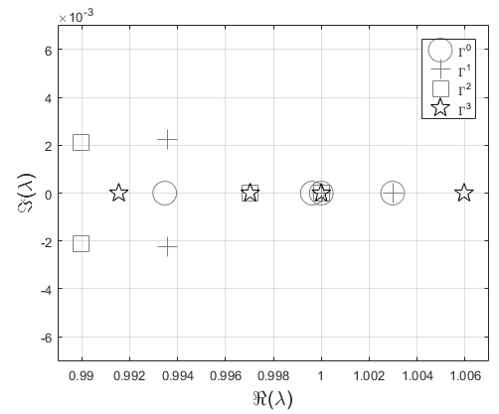}
\caption{\footnotesize{Eigenvalues represented in complex plane. Different symbols are used to distinguish the eigenvalues of the four steady-states.}}\label{fig:Eigenvalues}
\end{figure}

\begin{table}
\footnotesize
\begin{center}
  \begin{tabular}{@{\extracolsep{\fill}} | c | c | c | }
  \hline
  \textbf{Steady-state} & \textbf{(h,x,y,z)} & \textbf{Eigenvalues} \\
  \hline
  \rule{0pt}{3ex}
  $\Gamma^0$  & (0,0,0,0)  & (1, 0.999996, 0.999934, 1.000030)  \\
  \hline
  \rule{0pt}{3ex}
  $\Gamma^1$  & (0.599525,0,0,0)  & (1, 0.999935 $\pm$ 0.000023 $i$, 1.000030)  \\
  \hline
  \rule{0pt}{3ex} 
  $\Gamma^2$  & (1.796894, 0, 30, 0)  & (1, 0.999897 $\pm$ 0.000022 $i$, 0.999970)    \\
  \hline
  \rule{0pt}{3ex} 
  $\Gamma^3$  & (0, 0, 30, 0)  & (1, 1.005961, 0.991550, 0.997031)    \\
  \hline
  \end{tabular}
\caption{Steady-states and eigenvalues of the \eqref{eq:AAL_timedisc_expl}.}
\end{center}
\label{tab:1}
\end{table}

\section{Stability analysis}

For nonhyperbolic fixed points the Hartman--Grobman is not applicable and the effort to infer about their stability is significantly greater. The centre manifold theory helps when some of the eigenvalues of the Jacobian have unitary absolute value and the others have absolute value less than one. It allows reducing the dimensionality of a multi-dimensional dynamical system around a nonhyperbolic fixed point and determine its stability properties by studying the centre manifold \cite{wiggins}. As visible in \figurename~\ref{fig:Eigenvalues} and Table 1, the steady-states $\Gamma^0,\Gamma^1$ and $\Gamma^3$ have one eigenvalue with absolute value falling outside the unitary circle, then the centre manifold theory is not applicable. In the present research the stability of these fixed points is not analytically shown but inferred by means the phase-portrait analysis (see the Subsection \ref{subsec:phaseportrait}). $\Gamma^2$ has one eigenvalue with unitary module and the others less than one, then the centre manifold theory is applicable. Nonetheless, from a strictly ecological point of view, the initial conditions leading to $\Gamma^2$ are the most realistic, then the effort to determinate the stability of $\Gamma^2$ is required.

\subsection{Centre Manifold Theory}
Following the approach of \cite{murakami}, consider the $m$-dimensional system
\begin{align}\label{eq:decomposit0}
X_{n+1}=AX_n+F(X_n),
\end{align} 
where $X_n\in\mathbb{R}^m$ and $F(X_n)=O(||X||^2)$. Let $\lambda$ an eigenvalue of $A$, $p\in\mathbb{R}^{1\times m}$ and $q\in\mathbb{R}^{m}$ the eigenvectors which satisfy 
\begin{align}
Aq=\lambda q,\,\,\,\,\,\,\,\,\,\,pA=\lambda p,\,\,\,\,\,\,\,\,\,\,pq=1.
\end{align}
Let $u=pX\in\mathbb{R}$ and $v=X-qu\in\mathbb{R}^m$, then is simple to show that $X$ can be decomposed as
\begin{align}\label{eq:decomposit}
X=qu+v.
\end{align}
If $A$ has only one eigenvalue such that $|\lambda|=1$ and the other less than one, there exists a function $v=G(u)$ such that $G(0)=0$ and $G'(0)=0$. 
Let the center manifold
\begin{align}\label{eq:v}
v=G(u)=C_2 u^2+C_3 u^3+\ldots,
\end{align}
and 
\begin{align}\label{eq:F}
F(qu+G(u))=F_2u^2+F_3u^3+\ldots,
\end{align}
where $C_2$ and $C_3$ are given by:
\begin{equation}\label{eq:coeff_C}
\left\{
\begin{array}{l}
C_2=(\lambda^2I-A)^{-1}(I-qp)F_2, \\
C_3=(\lambda^3I-A)^{-1}\left((I-qp)F_3-2\lambda pF_2C_2\right), \\
\vdots
\end{array}
\right.
\end{equation}
Substituting \eqref{eq:v} in \eqref{eq:decomposit}:
\begin{align}
X=qu+C_2 u^2+C_3 u^3+\ldots,
\end{align}
then,
\begin{align}\label{eq:centre_manifold}
X=qu+(\lambda^2I-A)^{-1}(I-qp)F_2 u^2+(\lambda^3I-A)^{-1}\big((I-qp)F_3-2\lambda pF_2C_2\big) u^3+\ldots,
\end{align}
in which $F_2$ and $F_3$ can be evaluated by comparison of \eqref{eq:decomposit0} and \eqref{eq:F}.
Given the centre manifold described by \eqref{eq:centre_manifold}, in order to evaluate the stability of the nonhyperbolic fixed point, it is possible to make use of the following theorem:
\begin{thm}[See \cite{Elaydi}]\label{th:nonhyp}
Let $\tilde\xi$ be a fixed point of a map $f$ such that $f'(\tilde\xi)=1$. If $f'(\tilde\xi)$, $f''(\tilde\xi)$, and $f'''(\tilde\xi)$ are continuous at $\tilde\xi$, then the following statements hold:
\begin{itemize}[noitemsep,topsep=1pt]
\item[\textbf{A.}] if $f''(\tilde\xi)\neq 0$, then $\tilde\xi$ is unstable;
\item[\textbf{B.}] if $f''(\tilde\xi)=0$ and $f'''(\tilde\xi)>0$ then $\tilde\xi$ is unstable;
\item[\textbf{C.}] if $f''(\tilde\xi)=0$ and $f'''(\tilde\xi)<0$ then $\tilde\xi$ is asymptotically stable.
\end{itemize}
\end{thm}

\subsection{Stability of the steady-state $\Gamma^2$}

In this part of the study only the most important steps were reported. The eigenvectors left $p$ and right $q$ associated to the unitary eigenvalue of $\Gamma^2$ are:
\begin{equation}
q=
\begin{pmatrix}
    1 \\
    0 \\
    0 \\
    0
\end{pmatrix}
,\,\,\, p=
\begin{pmatrix}
    1, 0, 0 ,0
\end{pmatrix}.
\end{equation}
The \eqref{eq:AAL_timedisc_expl} was first rewritten by evaluating the Taylor series around $\Gamma^2$, in order to simplify the equations. 
Making the change of variables
\begin{align}
\left(\hat h,\hat x,\hat y,\hat z\right)\longrightarrow \left(\hat h-h,\hat x-x,\hat y-y,\hat z-z\right),
\end{align}
the fixed point was shifted to the origin. The set of difference equations can be written
\begin{equation}\label{eq:AAL_Taylor_0}
\left\{
\begin{array}{l}
\hat h_{n+1}=\hat\Phi_1(\hat h_n,\hat x_n,\hat y_n,\hat z_n) \\

\hat x_{n+1}=\hat\Phi_2(\hat h_n,\hat x_n,\hat y_n,\hat z_n) \\

\hat y_{n+1}=\hat\Phi_3(\hat h_n,\hat x_n,\hat y_n,\hat z_n) \\

\hat z_{n+1}=\hat\Phi_4(\hat h_n,\hat x_n,\hat y_n,\hat z_n), \\
\end{array}
\right.
\end{equation} 
in which $\hat\Phi_i(\hat h_n,\hat x_n,\hat y_n,\hat z_n)$, $i=1,2,3,4$, expresses the function $\Phi(h_n,x_n,y_n,z_n)$ after both Taylor series expansion and the change of variable. The system can be rewritten 
\begin{equation}
\begin{pmatrix}
    \hat h_{n+1} \\
    \hat x_{n+1} \\
    \hat y_{n+1} \\
    \hat z_{n+1}
\end{pmatrix}
= A 
\begin{pmatrix}
    \hat h_{n} \\
    \hat x_{n} \\
    \hat y_{n} \\
    \hat z_{n}
\end{pmatrix} 
+ \begin{pmatrix}
    F_1(\hat h_n,\hat x_n,\hat y_n,\hat z_n) \\
    F_2(\hat h_n,\hat x_n,\hat y_n,\hat z_n) \\
    F_3(\hat h_n,\hat x_n,\hat y_n,\hat z_n) \\
    F_4(\hat h_n,\hat x_n,\hat y_n,\hat z_n)
\end{pmatrix}, 
\end{equation}
in which $A$ is a $4\times 4$ matrix containing the coefficients of the linear terms and $F_i(\hat h_n,\hat x_n,\hat y_n,\hat z_n)$ represent the nonlinear ones. By using the same parameters values adopted in Subsection \ref{subsec:steady-states}, the matrix A is:

\begin{align}
A=\left(
\begin{array}{cccc}
 1 & 5\cdot 10^{-11} & 0 & 0 \\
 0 & 1.00003 & 0 & 6.25\cdot 10^{-7} \\
 0 & -5.99\cdot 10^{-4} & 1.00001 & 0 \\
 0 & 6.25\cdot 10^{-7} & 0 & 0.999991 \\
\end{array}
\right).
\end{align}
With regard to $F(\hat h_n,\hat x_n,\hat y_n,\hat z_n)$, by neglecting the terms with coefficients less than $10^{-7}$, it is:
\begin{align}
F=\left(
\begin{array}{c}
 0 \\
 -2.81\cdot 10^{-3} \hat h \hat x^2+1500.06 \hat h \hat x \hat z- 10^{-5} \hat h \hat x+93.753 \hat x^2-7.5\cdot 10^7 \hat x \hat z \\
 -6.67\cdot 10^{-5} \hat x^2 \hat y+0.002\hat x^2+2\cdot 10^{-5} \hat x \hat y \\
1.4\cdot 10^{-5} \hat x^2 \hat y-93.75 \hat x^2-7.5\hat x \hat y \hat z-1.875\cdot 10^{20} \hat x \hat z^2+7.5\cdot 10^7 \hat x \hat z+93.75 \hat z^2 \\
\end{array}
\right).
\end{align}
Taking into account the first equation of \eqref{eq:AAL_Taylor_0},
\begin{align}
h_{n+1}=\hat\Phi(\hat h_n,\hat x_n,\hat y_n,\hat z_n),
\end{align}
it is possible to show that the centre manifold equation is:
\begin{align}
f(u)=u-0.01782u^3,
\end{align}
that respects the conditions $f''(\tilde u)=0,\,\,\,f'''(\tilde u)<0$, implying that $\Gamma^2$ is asymptotically stable.

A necessary, but not sufficient, condition for bifurcation of a fixed point is that the fixed point is nonhyperbolic \cite{wiggins}; in the present research the bifurcation problem was not addressed due to the complexity of the equations.

\section{Numerical simulations}

\subsection{Phase-space analysis}\label{subsec:phaseportrait}

\begin{figure}[!ht]
\centering
\includegraphics[width=0.9\textwidth]{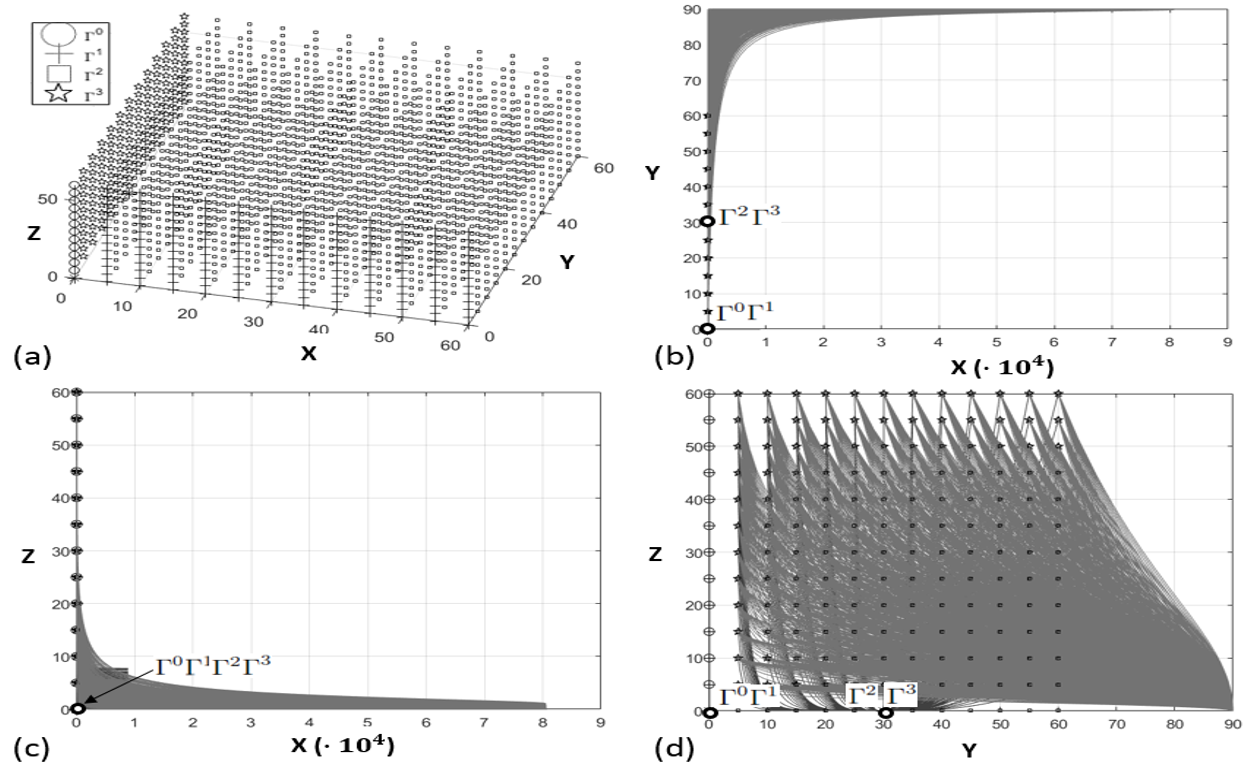}
\caption{\footnotesize{\textbf{(a)} Initial conditions $x_0,y_0,z_0$, assuming $h_0=0$, used to compute the phase-portraits shown in figures \textbf{a}, \textbf{b} and \textbf{c}. Different symbols were used to identify the steady-state they lead to according to the legend (e.g. the initial conditions identified by the squares lead to $\Gamma^2$ steady-state). The plots \textbf{b}, \textbf{c} and \textbf{d} represent the phase-portraits of respectively $x$-$y$, $z$-$x$ and $z$-$x$ populations (note that the $x$-axis is not in scale with $y$ and $z$ axes). The parameters used are: $r_1=r_2=0.3$, $a=5\cdot 10^{-6}$, $k=5\cdot 10^{-4}$, $\nu=1$, $b=0$, $p=1$, $\alpha_{12}=\alpha_{21}=0.6$, $d_2=0.01$, $\epsilon_1=\epsilon_2=0.3$; the discretization parameters are $dt=10^{-5}$ and $q=1.5$. Each steady-state was labeled with a black-white circle.}}\label{fig:phase-space}
\end{figure}

The phase-space analysis of the system \eqref{eq:AAL_timedisc_expl} was carried out and, to facilitate the visualization, three phase-portraits, containing the $x$-$y$, $x$-$z$ and $y$-$z$ populations, are represented respectively in figures \ref{fig:phase-space}\textbf{b},\textbf{c} and \textbf{d}. Each phase-portrait was calculated by considering the initial conditions displayed in \figurename~\ref{fig:phase-space}\textbf{a}: $h_0=0$, $0\le x,y,z\le 60$. The initial conditions were reported also in the three phase-portraits, however, in plots \textbf{b} and \textbf{c}, the $x$-axis is in $10^4$ units due to the important range of variation of the aphids population, then the range of variation of $x$ initial conditions is not recognizable. Each initial condition point of \figurename~\ref{fig:phase-space}\textbf{a} is plotted with a symbol indicating the steady-state it leads to (i.e. the points characterized by $x=0,y=0,0\le z\le 60$ are marked by an empty circle and lead to the steady-state $\Gamma^0$, according to the legend).

Note that where the steady-states are superimposed each other it means that the variables considered in that phase-portrait do not make the difference, e.g., in $x$-$z$ plane represented in \figurename~\ref{fig:phase-space}\textbf{b}, $\Gamma^2$ and $\Gamma^3$ are coincident since is the variable $h$ (not shown) that makes the difference.

About the phase-space analysis, it is possible to deduce that: 
\begin{itemize}
\item the steady-state $\Gamma^0$ is reached if initial populations are characterized by $x_0=y_0=0$; $\Gamma^1$ is reached if $y_0=0$, regardless of $x_0$ and $z_0$; 
\item $\Gamma^3$ is reached if $x_0=0$, regardless of $y_0$ and $z_0$; 
\item the system evolves toward $\Gamma^2$ for any other initial condition. In this case, the phase-portraits show that the trajectories experience a peak of the $x$ and $y$ populations and a convergence towards $\Gamma^2$. 
\end{itemize}
Numerical values of the maximum populations will be given in the next Section.

\subsection{Time series analysis}

\begin{figure}[!ht]
\centering
\includegraphics[width=\textwidth]{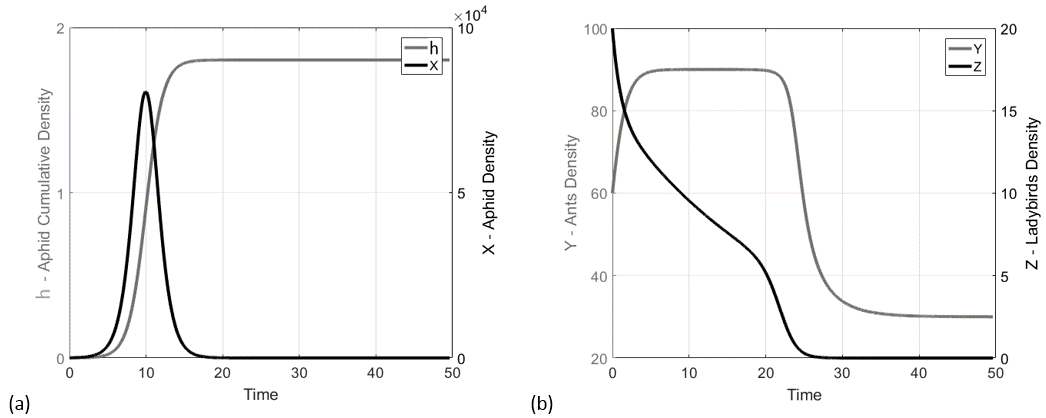}
\caption{\footnotesize{Trajectories of $h,x,y,z$ obtained by iteration of \eqref{eq:AAL_timedisc_expl}, by assuming the same parameters values of \figurename~\ref{fig:phase-space}:\textbf{(a)} $h$ is represented by the gray line and $x$ by the black one; \textbf{(b)} $y$ is represented by the gray line and $z$ by the black one.}}\label{fig:trajectories}
\end{figure}

\begin{figure}[!ht]
\centering
\includegraphics[width=\textwidth]{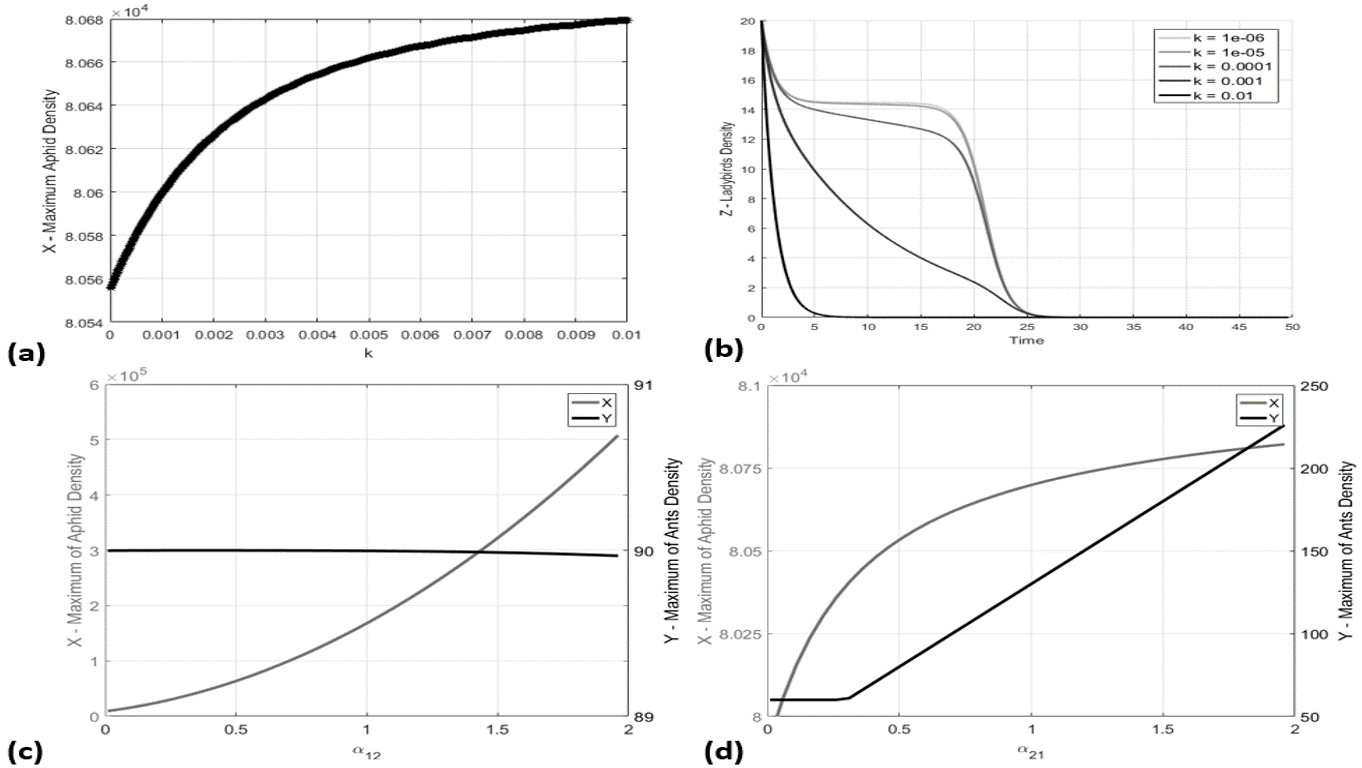}
\caption{\footnotesize{\textbf{(a)} Maximum of the aphid density versus the coefficient of interspecific competition; \textbf{(b)} Evolution of the ladybird density, evaluated for $k=10^{-2},10^{-3},10^{-4},10^{-5},10^{-6}$ (progressively darker colors are assigned to greater $k$ values). \textbf{(c)} Maximum of aphids (gray) and ants (black) density versus $\alpha_{12}$. \textbf{(d)} Maximum of aphid (gray) and ants (black) density versus $\alpha_{21}$.}}\label{fig:k_alpha-parameters}
\end{figure}

The trajectories produced by the system \eqref{eq:AAL_timedisc_expl} were depicted in \figurename~\ref{fig:trajectories} by using the initial conditions $h_0=0,x_0=y_0=60,z_0=20$. The trajectories for the couples $h,x$ and $y,z$ were represented respectively in \figurename~\ref{fig:trajectories}\textbf{a} and \textbf{b}. As we can see, the system evolves towards the steady-state $\Gamma^2$. The aphid density grows up to a maximum of $8058$ individuals at time $t=9.96$ and drops to 0 individuals for $t\gtrsim 20$. As a consequence of the aphids-ants mutualistic relationship, the ants increase their density population up to 90 individuals in the period of greatest abundance of aphids, although they must face a sharp decline and a stabilization on $30$ individuals after the density population of aphids collapses. The ladybird population density results monotonic decreasing, becoming zero together with the aphids. 

Nonetheless, the shape of their population curve is sensitively influenced by the $k$ coefficient of interspecific competition between ants and ladybirds. In \figurename~\ref{fig:k_alpha-parameters}\textbf{b} the $z$ population was calculated by assuming $k$ in range $10^{-5}$--$10^{-2}$. It is possible to see that lower $k$ values make the ladybird population supported by the aphid colony size. In particular, with $k=10^{-5}$--$10^{-6}$ the ladybird colony is almost constant during the period of abundance of the aphids, while greater values imply an increasingly rapid decline of the ladybirds population size. The maximum of the aphids population density is also affected by the magnitude of $k$, although there is a variation of only 100 individuals about between $k=10^{-5}$ and $10^{-2}$, see \figurename~\ref{fig:k_alpha-parameters}\textbf{a}. 

The influence of the parameter $\alpha_{12}$, quantifying the advantages that the presence of ants induces on the growth of aphids, was also studied. From the curves represented in \figurename~\ref{fig:k_alpha-parameters}\textbf{c} it is possible to conclude that the maximum of the aphids abundance sensitively increases with $\alpha_{12}$, while this parameter scarcely influences the ants density. Conversely, the parameter $\alpha_{21}$, quantifying the advantages that the presence of aphids induces on the growth of ants, sensitively influences the growth of the ant colony, determining only a small percentage growth of the aphid colony, as visible in \figurename~\ref{fig:k_alpha-parameters}\textbf{d}. 

\section{Conclusions}

By taking advantage of the mathematical model proposed by Kindlmann \& Dixon in 2003 \cite{Kindlmann2003} to describe the ladybirds-aphids (prey-predator) population density and the model built by Holland \& De Angelis in 2010 \cite{HollandDeAngelis2010} describing the population dynamics of two mutualistic species, in this research a new set of differential equations was proposed to describe the within-season population dynamics of an ecological patch hosting a community of ladybirds, aphids and ants. This framework describes the population dynamics in presence of both prey-predator and mutualistic relationships, based on the following facts: 
\begin{itemize}
\item between aphids and ants exists a well documented mutualistic relationship; 
\item ladybirds behaves as predator of aphids; 
\item ladybirds and ants may experience an interspecific competition.
\end{itemize}

The proposed model, first formulated in continuous-time, then in discrete-time domain by taking advantage of the NSFD scheme, is characterized by four nonhyperbolic steady-states $\Gamma^0$-$\Gamma^3$. The phase-space analysis highlighted that the initial conditions most realistic from an ecological point of view lead to the steady-state $\Gamma^2$, characterized by a collapsed colony of aphids and ladybirds and a population of ants equal to $r_2/d_2$ ($r_2$ is their maximum potential growth rate and $d_2$ their self-limiting term). The asymptotic stability of $\Gamma^2$ was also shown via the centre manifold theory. It is an feature that the mutualistic relationship does not influence any of the asymptotic population densities.

The time-series analysis shown that, with sufficiently large initial populations, the populations of aphids and ladybirds reach a maximum, then definitely collapse according to the $\Gamma^2$ population densities. Then, the presence of aphids and ladybirds constitutes a transitory phenomenon. Nonetheless, the presence of the mutualistic relationship determine a sensitive grows of the aphids and ants population densities peaks. 

Since the four steady-states are nonhyperbolic, future research activities should be carried out to investigate the presence of bifurcations. Furthermore, a validation with experimental data, actually not publicly available to the author's knowledge, should be carried out in order to validate the reliability of the AAL model in predicting the populations evolution.
 
\section*{Acknowledgments}
I would like to thank Professors W. Nazeer, M. Imran and to an anonymous reviewer for giving me precious suggestions addressed to improve the quality of the paper.

\appendix

\newpage

\bibliographystyle{plain}

\label{lastpage}
\end{document}